\documentclass[12pt]{article}

\makeatletter
\def\@seccntformat#1{\csname the#1\endcsname.\hspace*{0.5em}}
\makeatother

\usepackage{amsmath}
\usepackage{amsfonts}
\usepackage{amssymb}
\usepackage{graphicx,subfigure}
\usepackage{url}
\usepackage{xspace}
\usepackage{amsbsy}%
\usepackage{epsfig}%

\usepackage{theorem}%
\begingroup \makeatletter
\gdef\th@plain{\normalfont\itshape
  \def\@begintheorem##1##2{%
        \item[\hskip\labelsep \theorem@headerfont ##1\ ##2.]}%
\def\@opargbegintheorem##1##2##3{%
   \item[\hskip\labelsep \theorem@headerfont ##1\ ##2\ (##3).]}}
\endgroup

\newcommand{\EKhref}[2]{URL: \href{#1}{\nolinkurl{#2}}}
\usepackage[%
breaklinks%
,colorlinks%
,linkcolor=red
,anchorcolor=blue%
,pagecolor=blue%
,citecolor=blue%
,bookmarks=false%
]{hyperref}
\usepackage{breakurl}

\newcommand{\comment}[1]{}

\newcommand{\bz}{{\mathbf 0}}

\newcommand{\ZZ}{{\mathbb Z}}

\newcommand{\RR}{{\mathbb R}}

\newcommand{\ff}{{\mathbf f}}

\newcommand{\xx}{{\mathbf x}}
\newcommand{\yy}{{\mathbf y}}

\newcommand{\vv}{{\mathbf v}}

\newtheorem{theorem}{Theorem}
\newtheorem{lemma}{Lemma}

\newtheorem{corollary}{Corollary}
\newtheorem{definition}{Definition}
\newtheorem{problem}{Problem}
{\theorembodyfont{\upshape} } 
{\theorembodyfont{\upshape} \newtheorem{example}{Example}} 

\makeatletter
\def\@yproof[#1]{\@proof{ #1}}
\def\@proof#1{\begin{trivlist}\item[]{\em Proof#1.}}
\newenvironment{proof}{\@ifnextchar[{\@yproof}{\@proof{}
}}{~$\Box$\end{trivlist}}
\makeatother

\title{%
Domain-of-Attraction Estimation for Uncertain Non-polynomial Systems\raisebox{0ex}{*}%
}

\chardef\ttlb="7B 
\chardef\ttrb="7D 
\chardef\ttti="7E 
\author{%
Min Wu$^a$, Zhengfeng Yang$^a$ and Wang Lin$^{b}$
\\
\\
\small\llap{$^a$}
Shanghai Key Laboratory of Trustworthy Computing\\
[-0.2ex] \small East China Normal University, Shanghai 200062, China
\\[-0.2ex]
\small\llap{$^b$}
College of Mathematics and Information Science \\
[-0.2ex] \small Wenzhou University, Zhejiang 325035, China
\\[-0.2ex]
\small {\ttfamily \{mwu,zfyang\}@sei.ecnu.edu.cn;
linwang@wzu.edu.cn}}
\begin{document}

\date{}
\maketitle
\def\thefootnote{\fnsymbol{footnote}}
\footnotetext[1]{%
\footnotesize
This material is supported in part by the National Natural Science
Foundation of China under Grants 91118007,61021004(Wu,Yang), and the
Fundamental Research Funds for the Central Universities under Grant
78210043(Wu,Yang). }

\begin{abstract}
In this paper, we consider the problem of
computing estimates of the domain-of-attraction for non-polynomial
systems. A
polynomial approximation technique, based on multivariate polynomial
interpolation and error analysis for remaining functions, 
 is applied to compute an uncertain
polynomial system,
whose set of trajectories contains that of the original
non-polynomial system.
Experiments on the benchmark non-polynomial systems show that our
approach gives better estimates of the
domain-of-attraction.
\end{abstract}

\section{Introduction}

Stability for nonlinear control systems plays an important role in control system analysis and design. It will be very useful to
know the
domain of attraction (DOA) of an equilibrium point,
 however, this
region is usually difficult to find and represent explicitly.
Therefore, looking for
underestimates of the DOA with simple shapes has been a fundamental
issue in control system analysis since a long time. Among all the
methods, those based on Lyapunov functions are dominant in
literature \cite{chesi2005domain, chesi2009estimating, chesi2005lmi,
CT1989, cruck2001estimation, jarvis2003lyapunov,
prakash2002obtaining, tan2006stability, topcu2010robust,
tibken2000estimation, hachicho2002estimating, tibken2002computation,
prajna2004nonlinear}. These methods
not only yield a Lyapunov function
as a stability certificate,
but also the corresponding sublevel sets
as
estimates of the DOA.

For polynomial systems, many well-established techniques
(\cite{chesi2005lmi, cruck2001estimation, CT1989,
jarvis2003lyapunov, prakash2002obtaining, tan2006stability,
topcu2010robust, tibken2000estimation, hachicho2002estimating,
tibken2002computation, prajna2004nonlinear}) are available for
computing estimates of DOAs.
In \cite{tan2006stability}, a method based on SOS
decomposition was presented to find provable
DOAs and attractive invariant sets for nonlinear
polynomial systems.
 For odd polynomial systems, \cite{chesi2005lmi} employed an LMI-based
method to compute the optimal quadratic Lyapunov function for
maximizing the volume of the largest estimate of the DOA.  To obtain estimates of DOAs of uncertain
polynomial
systems,
the authors of \cite{cruck2001estimation} used discretization (in time) to
flow invariant sets backwards along the flow of the vector field.
In \cite{prakash2002obtaining},  quantifier elimination (QE) method via
QEPCAD was also applied to find Lyapunov functions for estimating
the DOA. However,
these methods cannot be applied directly in practice since most real
systems are non-polynomial systems,
i.e, their vector fields contain non-polynomial terms. For
this kind of systems, only a few approaches have been proposed to deal
with the DOA analysis. In \cite{chesi2005domain, chesi2009estimating,chesi2011domain}, the author
proposed an LMI technique through
Taylor expansions as substitution for non-polynomial terms, and this technique can be generalized to compute estimates of DOAs for uncertain non-polynomial systems. In
\cite{warthenpfuhl2010interval}, an interval arithmetic approach was proposed. Recently, \cite{saleme2011estimation, saleme2012new} suggested
a new method, based on quadratic Lyapunov function and the theorem
of Ehlich and Zeller.

In this paper, we will consider the problem of stability region analysis of uncertain non-polynomial systems.
Through multivariate polynomial interpolation together with the
interpolation
error analysis,
 we substitute a non-polynomial system as an uncertain polynomial system,
whose set of trajectories contains that of the original non-polynomial system.
By computing estimates of the DOA for the resulted uncertain polynomial system, we obtain estimates of the DOA for the original
non-polynomial system.
Our method is also applicable to the problem of searching for the
largest possible underestimate of the DOA via a fixed Lyapunov function.
Compared with the classical approximation by Taylor expansions,
the error bound obtained using our suggested
method is much sharper, which helps to yield a larger
estimate of the DOA for a given
non-polynomial system.

The rest of the paper is organized as follows. In
Section~\ref{sect:Pre}, some notions related to DOAs are presented.
In Section \ref{sec:approx}, a
polynomial approximation method, based on multivariate
polynomial interpolation and
interpolation
error analysis, is proposed to
substitute the non-polynomial functions
as uncertain polynomials. In Section \ref{sec:ROA},
bilinear SOS programming is applied to estimate DOAs of non-polynomial systems. In
Section~\ref{sec:exp}, experiments on some benchmarks are shown to
illustrate our suggested method. Section~\ref{sect:conclusion}
concludes the paper.


\section{Problem Formulation}\label{sect:Pre}


Consider
an autonomous system
\begin{equation}\label{EQ:nonlinear}
\dot \xx = \ff(\xx),
\end{equation}
where $\ff:D\subseteq\RR^n \rightarrow \RR^n$
is a continuous function defined on an open set $D$ and $\ff$
satisfies the Lipschitz condition:
$$
\|\ff(\xx)-\ff(\yy) \|\le L\| \xx-\yy\|\quad\mbox{for all $\xx,\yy\in D$.}
$$
Denote by $\phi(t;\xx_0)$ the
solution of (\ref{EQ:nonlinear}) with the given initial value  $\xx(0)=\xx_0$.

A vector $\xx\in \RR^n$ is an {\em equilibrium} point of the
system~(\ref{EQ:nonlinear}) if $\ff(\xx) =\bz$.
Since any equilibrium point can be shifted to the origin  $\bz$ via
a change of variables, we may
assume without loss of generality that the equilibrium point of interest occurs at the origin.
The equilibrium point $\bz$ of
(\ref{EQ:nonlinear}) is said to be {\em stable}, if for any $\epsilon >0$ there exists  $\delta$ such that whenever $\|\xx_0\|<\delta$ we have
    $\|\phi(t;\xx_0)\|<\epsilon$  for all $t > 0;$
the point $\bz$ is said to be {\em unstable} if it is not stable; $\bz$ is {\em asymptotically stable}, if, in addition to being stable, there exists $\delta$ such that $ \lim_{t\rightarrow \infty} \phi(t;\xx_0)=\bz$ whenever $\|\xx_0\|<\delta$;  the equilibrium point
    $\bz$ is {\em globally asymptotically stable}, if, in addition to being stable, we have $\lim_{t\rightarrow \infty} \phi(t;\xx_0)=\bz$  for all $\xx_0\in\RR^n$.

Globally asymptotic stability is very desirable
but
is
usually difficult to achieve. When the
equilibrium point $\bz$ is asymptotically stable, we are interested in
determining how far the trajectory of (\ref{EQ:nonlinear}) can be from $\bz$ and still
converge to $\bz$ as $t$ approaches $\infty$. This gives rise
to the following definition.

\begin{definition}[Domain of Attraction]\label{def:roa}
The {\em domain of attraction} (DOA)
of the
equilibrium point
$\bz$ for the system (\ref{EQ:nonlinear})
is defined to be the set $
\{ \xx\in \RR^n |
\lim_{t\rightarrow \infty} \phi(t; \xx)=\bz\}. $
\end{definition}

Usually, no algebraic description for
DOAs is available. So researchers are mainly concerned with
computing underestimates of the DOAs. Many well-established
techniques (\cite{chesi2005lmi, cruck2001estimation, CT1989,
jarvis2003lyapunov, prakash2002obtaining, tan2006stability,
topcu2010robust, tibken2000estimation, hachicho2002estimating,
tibken2002computation, prajna2004nonlinear}) are available for
computing estimates of DOAs for polynomial (control) systems, i.e.,  autonomous
systems with polynomial vector fields.
However, in practice, many autonomous systems often contain non-polynomial terms in their vector fields. Below is an example.
\begin{example}\cite[Example 1.2.1]{khalil2002nonlinear} %
Consider the simple pendulum shown in Figure 1. The motion of the
pendulum is described 
by the following equation
\begin{figure}
\centering
\includegraphics[width=0.35\textwidth]{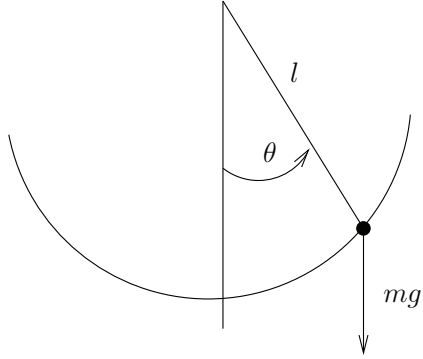} \caption{\small Pendulum.}
\end{figure}
$$ml\ddot{\theta} =-mg\sin\theta-kl\dot{\theta},$$
where $\theta$ denotes the angle subtended by the rod and the
vertical axis through the pivot point,~$l$ the length of the
rod, $m$ the mass of the bob, $g$ the acceleration due
to gravity, and~$k$ the coefficient of friction. Let us take the state variables
as $x_1=\theta$, and $x_2=\dot{\theta}$. Then
the above equation is converted into a non-polynomial system
\begin{eqnarray*}
\left\{\begin{array}{l@{}l}
   \dot{x}_1=x_2, \\
   \dot{x}_2=-\frac{g}{l}\sin x_1-\frac{k}{m}x_2.
   \end{array}\right.
\end{eqnarray*}
\vspace{-.2cm} $ \hfill \Box$
\end{example}

For the case of non-polynomial
(control) systems, the
problem
of computing DOAs
is still open, and
only a few approaches have been proposed to deal
with stability region analysis:
in \cite{chesi2005domain,
chesi2009estimating, chesi2011domain},
the authors suggested
a  way
to approximate
non-polynomial
vector fields by Taylor series expansion
at
the origin; in \cite{warthenpfuhl2010interval}, an interval
arithmetic approach for the estimation of the DOA
was
proposed;
and recently, a method based on the theorem
by
 Ehlich and Zeller was presented in
\cite{saleme2011estimation, saleme2012new}. In this paper,
we will apply polynomial
approximation to transform
a non-polynomial system into
an uncertain polynomial system,
whose
set of trajectories contains that of the original non-polynomial system. Therefore, underestimate estimates of the DOA of the latter system yield those for the original
non-polynomial
system.


\section{Polynomial Approximation}\label{sec:approx}

A key problem in estimating the DOA of a non-polynomial system is how
to approximate the involved non-polynomial terms using polynomials,
yielding
an uncertain polynomial
system
with the equilibrium $\bz$ being kept. This problem is
further reduced to the following problem.
\begin{problem}
Let $\phi(\xx):\Psi\rightarrow \RR $ be a non-polynomial function where $\Psi\subset \RR^n$ is a bounded subset containing the origin $\bz$.
Given
$d\in\ZZ_{\ge 0}$, we will find a polynomial $p(\xx)$ with degree $d$
such that the
error function $r_{d}(\xx)=\phi(\xx)-p(\xx)$ satisfies
   $r_{d}(\bz)=0$ and the value $\max_{\xx\in \Psi} |r_d(\xx)|$ is minimized.
\end{problem}
The classic method of polynomial approximation is
Taylor expansions.
%
Suppose $\phi(\xx)$ is a $d$ times continuously
differentiable in $\Psi$.
The Taylor expansion of $\phi(\xx)$ at the origin $\bz$  is
\begin{equation*}\label{taylor_exp}
\phi(\xx)=
\underbrace{\sum_{|\alpha|\leq d-1}
  \frac{D^{\alpha}\phi(\bz)}{\alpha!}\xx^\alpha}_{p(\xx)}+\underbrace{\sum_{|\beta|=d}
  \frac{D^{\beta}\phi(\xi)}{\beta!}
  \xx^\beta}_{r_d(\xx)}
\end{equation*}
for some $\xi\in(0,\xx)$.
In the above expression,
$p(\xx)$ is an approximate polynomial of $\phi(\xx)$ and the remainder term $r_d(\xx)$ is the
error function of this approximation.
Clearly, if the size of the region $\Psi$ is small enough, the
above
Taylor expansion
yields a tight bound of
$r_d(\xx)
$
for all $\xx\in \Psi$.
However,
when the size of $\Psi$ is large, the associated
error
bound may be too
loose.

To obtain a tighter bound, we will apply multivariate polynomial
interpolation (\cite{gasca2000history}) to compute an approximate
polynomial $p(\xx)$ of $\phi(\xx)$ with a given degree $d$.
%
Fix the graded lexicographic order
in $\RR[\xx]$. For the function~$\phi(\xx)$, one may find the minimal
monomial
$\xx^{\gamma}$ with $
\gamma:=
(\gamma_{1},\ldots,\gamma_{n}) \in {\ZZ}_{\geq 0}^{n}$,
such that
$\lim_{\xx\rightarrow\bz}\frac{\phi(\xx)-\phi(\bz)}{\xx^\gamma}\neq 0.$
Set
$\psi(\xx)=\frac{\phi(\xx)-\phi(\bz)}{\xx^\gamma}.$
Let $d\in \ZZ_{\ge 0}$ be such that $d> |\gamma|:=\sum_{i=1}^n \gamma_i$.
%
%
We
construct  a mesh~$M$ on~$\Psi$ with mesh spacing $s\in\RR_+$ and
mesh points set $\chi=\{
 \vv_1, \vv_2, \dots, \vv_k\}$
where
$k=\begin{pmatrix}n+d-|\gamma|\\ n\end{pmatrix}$.
Like in \cite{asarin2003reachability}, the meshes in our paper are
either rectangular or simplicial. Then, we apply Lagrange
interpolation to construct a polynomial~$\tilde{p}(\xx)$  as an
approximation of $\psi(\xx)$ through the interpolation points
$\chi$, i.e., $
   \tilde{p}(\vv_i)=\psi(\vv_i), \mbox{for } i=1,\dots,k.
$
%
%
Next, we will compute a tight bound  of
the interpolation error function
$ \tilde{r}(\xx):=\psi(\xx)-\tilde{p}(\xx)$.
Our idea is based on the following lemma.
\begin{lemma}\label{thm:Zeng}
\cite[Theorem 3]{zeng2010} Let $K\subset\RR^n$ be a convex
polyhedron, and $\vv_1,\vv_2,\dots,\vv_k$ and $s$ be the vertices and
diameter of $K$ respectively. Suppose that $\varphi: K\rightarrow\RR$ is
a continuous and differential function on $K$, and
$\lambda=\sup_{\xx\in K} \|\bigtriangledown \varphi(\xx)\|$. Then
for all $a_1,a_2,\dots,a_k\in\RR_+$ such that $ a_1+a_2+\cdots+a_k=1$, we have
$$
  |\varphi(\xx)-(a_1\varphi(\vv_1)+a_2\varphi(\vv_2)+\cdots+a_k\varphi(\vv_k))|\leq \frac{n}{n+1}\lambda s.
$$
\end{lemma}

The following corollary gives an estimated bound of $\tilde r(\xx)$ for $\xx\in \Psi \cap M$.

\begin{corollary}\label{thm:bound_error}
Let $s$ and $\chi :=\{ \vv_1,  \vv_2, \dots,  \vv_k \}$ be the mesh
spacing and mesh points set of~$M$, respectively. Suppose that
$\tilde{p}(\xx)$ is the interpolation polynomial of $\psi(\xx)$
through $\chi$, and $ \tilde{r}(\xx) =\psi(\xx)-\tilde{p}(\xx)$ is
the corresponding error function.
Let $\lambda=\sup_{\xx\in
\Psi\cap M} \|\bigtriangledown \tilde{r}(\xx)\|$.
Then
\begin{equation*}
  |\tilde{r}(\xx)|\leq  \frac{n}{n+1} \lambda s\quad\mbox{ for all $\xx\in M$}.
\end{equation*}
\end{corollary}
\begin{proof}
Clearly, $\tilde{r}(\xx)$ is a continuous and differential
function on $M$, and
$$
\tilde{r}(\vv_1)=\tilde{r}(\vv_2)=\cdots=\tilde{r}(\vv_k)=0.
$$
 Thus,
according to Lemma \ref{thm:Zeng}, for all
$a_1,a_2,\dots,a_v\in\RR_+$ such that $ a_1+a_2+\cdots+a_v=1$,
\begin{eqnarray*}
\begin{split}
|\tilde{r}(\xx)-(a_1\tilde{r}(\vv_1)+a_2\tilde{r}(\vv_2)+\cdots+a_k\tilde{r}(\vv_k))|
=|\tilde{r}(\xx)|\leq\frac{n}{n+1}\lambda s.
\end{split}
\end{eqnarray*}
\end{proof}

Therefore,
a non-polynomial function $\phi(\xx)$ can be relaxed to an uncertain polynomial,
as shown in the following theorem.
\begin{theorem}\label{lm:inclusion}
For a non-polynomial function $\phi(\xx)$ with $\xx\in \Psi$, let
$\xx^{\gamma}$ with $\gamma \in {\ZZ}_{\geq 0}^{n}$ be the minimal
monomial such that
$\lim_{\xx\rightarrow\bz}\frac{\phi(\xx)-\phi(\bz)}{\xx^\gamma}\neq
0$. Let $M$ be a mesh on $\Psi$ with mesh spacing~$s\in\RR_+$ and
the mesh point set $\chi=\{\vv_1, \vv_2, \dots, \vv_k\}$,
in which $k=\begin{pmatrix}n+d-|\gamma|\\
n\end{pmatrix}$. Suppose that $\tilde{p}(\xx)$ is  the interpolation
polynomial of $\psi(\xx):=\frac{\phi(\xx)-\phi(\bz)}{\xx^\gamma}$ at
$\chi$
with degree $
\leq d-|\gamma|$, $\tilde{r}(\xx)$ is the associated
interpolation
error function, and
$\lambda=\sup_{\xx\in 
\Psi\cap M} \|\bigtriangledown \tilde{r}(\xx)\|$.
Then for each $\xx \in
\Psi\cap M$
we have
$\phi(\xx)=p(\xx)+r_d(\xx),$  where
\begin{equation}\label{inclusion2}
p(\xx)=\phi(\bz)+\tilde{p}(\xx)\,\xx^{\gamma} \text{ and }
r_{d}(\xx)=u \,\xx^{\gamma}\mbox{ with } |u|\leq
\frac{n}{n+1}\lambda s.
\end{equation}
\end{theorem}
%
Clearly, the bound of the error $r_d(\xx)$ in (\ref{inclusion2})
depends on the mesh spacing $s$, which can yield a tighter bound.
Furthermore,  
the bound of $r_d(\xx)$ in (\ref{inclusion2}) will converge to zero
if $d\rightarrow\infty$.

\begin{example}\label{ex2}
Consider the function $\phi(x)=\cos x$ with $x\in\Psi=[-1.2, 1.2]$.
We want to compute a polynomial approximation for $\cos x$. Based on
Theorem~\ref{lm:inclusion}, we can obtain an uncertain
polynomial 
with degree $6$,
where
\begin{equation*} \left.\begin{array}{l@{}lr}
&p(x)=1-0.5x^2+0.0416525x^4-0.00134386x^6, & \\
&r_6(x)=ux^2, \,\, -0.0000336\leq u\leq 0.0000336. & \Box
\end{array}\right.
\end{equation*}
\end{example}

\section{Computation of Domain of Attraction}\label{sec:ROA}
In this section, we will consider an uncertain non-polynomial system of the form:
\begin{equation}\label{eq:nonpolynomial}
 \dot{\xx}=\ff(\xx,\theta)
\quad \mbox{ for all }\theta \in \Theta\subset \RR^t,
\end{equation}
where $\theta$ denotes
a vector of uncertainty.
Assume that the equilibrium point of
interest occurs at the origin $\bz$, i.e, $\ff(\bz,\theta)=\bz$ for all
$\theta\in\Theta$.
Denote by $\phi(t;\xx_0,\theta)$ the solution of (\ref{eq:nonpolynomial}) for the initial value  $\xx(0)=\xx_0$ and the uncertainty $\theta$.
The {\em Domain of Attraction} (DOA) of the system (\ref{eq:nonpolynomial}) is
defined as
$$
  \{ \xx\in \RR^n |
\lim_{t\rightarrow \infty} \phi(t; \xx,
\theta
)=\bz \quad\mbox{for all }
\theta\in\Theta\}.
$$
Lemma 1 in \cite{tan2006stability} can be modified a bit to compute underestimates of the DOA for (\ref{eq:nonpolynomial}) through Lyapunov functions, as described in the following theorem.
\begin{theorem}\label{thm:ROA}
\cite[Proposition 2.1]{topcu2010robust}
If there exists a continuously differentiable
function $V:\RR^n\rightarrow \RR$  such that
\begin{equation*}\label{condition:variable}
\left\{\begin{array}{l@{}l}
&  \Omega_V:=\{\xx\in\RR^n: V(\xx)\leq 1\} \text{ is bounded}, \\
&  V({\bz})=0,\\
&  V(\xx)>0,  \,\,\,\,\forall \xx\in \Omega_V\backslash\{\bz\},\\
& \dot{V}(\xx)=\frac{\partial V}{\partial \xx}\cdot \ff(\xx,\theta)< 0,
\,\,\,\, \forall \xx\in \Omega_V\backslash\{\bz\},\,\forall\theta\in\Theta
\end{array}\right.
\end{equation*}
then  $\Omega_V$
is an invariant subset of the DOA.
\end{theorem}

%
When the equilibrium $\bz$ is asymptotically stable, the set $\Omega_c$ is clearly an underestimate of the DOA since every trajectory starting in $\Omega_c$
remains in $\Omega_c$ and approaches $\bz$ as $t\rightarrow
\infty$. And, if the equilibrium $\bz$ is globally asymptotically stable then the DOA will be the whole space~$\RR^n$.

To enlarge the estimate
$\Omega_V$ given in Theorem \ref{thm:ROA},
\cite{tan2006stability} defined a variable sized region
$$
P_\beta=\{\xx\in\RR^n: g(\xx)\leq \beta\}
$$
with $g(\xx)$ a fixed and positive definite
polynomial in $\RR[\xx]$, for instance, $g(\xx)= \sum_{i=1}^n x_i^2$, and maximize $\beta$ subject to the constraint
$P_\beta\subseteq \Omega_V$ and the constraints in Theorem \ref{thm:ROA}.  Thus, the problem of computing
$\Omega_V$ can be transformed into the following problem:
\begin{eqnarray}\label{ROA:variable}
\left.\begin{array}{l@{}l}
    & \displaystyle \max  \beta   \\
    & \text{s.t.} \,\,\,\Omega_V \text{ is bounded}, \\
    & \quad\,\,\, V(\bz)=0,\\
    & \quad\,\,\, V(\xx)>0, \,\,\,\,\forall \xx\in \Omega_V\backslash\{\bz\},\\
    & \quad\,\,\,\dot V(\xx)=\frac{\partial V}{\partial \xx}\cdot\ff(\xx,\theta)<0, \,\,\,\,\forall \xx\in \Omega_V\backslash\{\bz\},\,\forall\theta\in\Theta\\
    &\quad\,\,\, g(\xx)\leq \beta\models V(\xx)\leq 1.
\end{array}\right\}
\end{eqnarray}

Suppose that the non-polynomial system (\ref{eq:nonpolynomial}) has the following form
\begin{equation}\label{eq:vectorfield}
 \dot{x}_i=f_i(\xx,\theta)=f_{i0}(\xx,\theta)+\sum_{j=1}^k f_{ij}(\xx,\theta)\phi_{ij}(\xx),\quad i=1,\dots, n,
\end{equation}
where $f_{ij}:\RR^n\times\RR^t \rightarrow \RR$ are
polynomials in $\xx$ for $j=0,\dots, k$,
and $\phi_{ij}:\RR^n \rightarrow \RR$
are non-polynomial
functions for $j=1,\dots, k$.
Using the polynomial approximation technique
in Section
\ref{sec:approx},
we can replace each non-polynomial term $\phi_{ij}(\xx)$
by an uncertain polynomial $p_{ij}(\xx)+u_{ij}\xx^{\gamma_{ij}}$ with the bound $|u_{ij}|\leq b_{ij}$.
This gives rise to
the
following uncertain polynomial system:
\begin{equation}\label{eq:uncertain}
\dot{x}_i=\hat{f}_i(\xx,\theta)=f_{i0}(\xx,\theta)+\sum_{j=1}^k
f_{ij}(\xx,\theta)(p_{ij}(\xx)+u_{ij}\xx^{\gamma_{ij}})
\end{equation}
where $|u_{ij}|\leq b_{ij}$ for  $i=1,\dots, n$.
It is not hard to prove that
the set of all trajectories of the system (\ref{eq:vectorfield}) is a subset of that of the system (\ref{eq:uncertain}),
and, consequently, the DOA of (\ref{eq:uncertain}) is actually  a subset of that of the system (\ref{eq:vectorfield}). Furthermore, the tighter the
bound $b_{ij}$ is, the closer
the DOA of (\ref{eq:uncertain}) is  to the DOA of the original system.

Next, we consider how to find an optimal estimate of the DOA
for the uncertain system~(\ref{eq:vectorfield})
through computing the $\Omega_V$ in (\ref{ROA:variable}).
Remark that the constraint $\dot V(\xx)=\frac{\partial V}{\partial \xx}\cdot\ff(\xx,\theta)<0$ in (\ref{eq:vectorfield})
may involve non-polynomial terms due to the existence of $\phi_{ij}(\xx)$'s. In this situation,
replacing the above constraint  by $\dot V(\xx)=\frac{\partial V}{\partial \xx}\cdot\hat\ff(\xx,\theta)<0$,
the problem (\ref{ROA:variable}) can be relaxed
as follow.
\begin{theorem}\label{thm:variable}
Let
$\hat \ff(\xx,\theta)=(\hat f_1(\xx,\theta),\dots, \hat f_n
(\xx,\theta))^T$ with $\hat f_i(\xx,\theta)$ given in
(\ref{eq:uncertain}). If $\widetilde{V}(\xx)$ is a solution of the
following polynomial optimization
\begin{eqnarray}\label{ROA:variable1}
\left.\begin{array}{l@{}l}
    & \displaystyle \max_{V, \beta}\  \beta   \\
    & \text{s.t.} \, V(\bz)=0,\\
    & \quad V(\xx)>0 \quad\forall \xx\in \Omega_V\backslash\{\bz\},\\
    & \,\,\,\frac{\partial V}{\partial \xx}\cdot\hat{\ff}(\xx,\theta)<0 \,\forall \xx\in \Omega_V\backslash\{\bz\},\forall\theta\in\Theta,\forall u_{ij}\in\{\pm b_{ij}\},\\
    & \quad g(\xx)\leq\beta\models V(\xx)\leq 1,
\end{array}\right\}
\end{eqnarray}
then $\Omega_{\widetilde{V}}:=\{\xx\in\RR^n: \widetilde{V}(\xx)\leq
1\}$ is an invariant subset of the DOA for (\ref{eq:uncertain}),
and therefore an invariant subset of DOA for (\ref{eq:vectorfield}).
\end{theorem}

\begin{proof}
By construction of $\hat f_i(\xx,\theta)$'s, we have
\begin{eqnarray*}\small
\begin{split}
  \forall\xx\in\Omega_{V}\backslash\{\bz\},\forall\theta\in\Theta,\,\,\exists
  \tilde{u}_{ij}\in[-b_{ij},b_{ij}]: \dot{V}(\xx)=\frac{\partial
  {V}}{\partial\xx}\cdot \hat{\ff}(\xx,\theta).
\end{split}
\end{eqnarray*}
Clearly, if the constraints in (\ref{ROA:variable1}) are fulfilled,
the conditions in (\ref{ROA:variable}) also hold. Therefore,
$\Omega_{\widetilde{V}}$ is certainly a subset of $\Omega_{V}$ in
(\ref{ROA:variable}).
\end{proof}

Assume that $\Theta$
is a semialgebraic set. For simplicity, we suppose $\Theta=\{\theta\in \Theta: \psi(\theta)\geq 0\}.$
By rewriting the third, fourth and fifth constraints into equivalent empty set conditions, the condition (\ref{ROA:variable1}) is transformed as
\begin{eqnarray}\label{condition:set}
\left\{\begin{array}{l@{}l}
    &  V(\bz)=0,\\
    & \{\xx\in\RR^n: V(\xx)\leq 1, \xx\neq\bz,V(\xx)\leq 0\}=\emptyset,\\
    & \{\xx\in\RR^n: V(\xx)\leq 1, \xx\neq\bz,\psi(\theta)\geq0,\frac{\partial V}{\partial \xx}\cdot\hat{\ff}(\xx)\geq 0\}=\emptyset,\forall u_{ij}\in\{\pm b_{ij}\},\\
    & \{\xx\in\RR^n: g(\xx)\leq\beta, V(\xx)\geq1,V(\xx)\neq1\}=\emptyset.
\end{array}\right.
\end{eqnarray}

As stated in \cite{tan2006stability}, Stengle's Positivstellensatz
\cite{bochnak1998real} be applied directly to solve
(\ref{condition:set}).
However, from the computational point of
view,
it is more efficient to replace
all the
inequations
in (\ref{condition:set}) by
inequalities of the form $f \geq 0$ or $f\leq 0$. This can be done
by introducing  constants $\delta\in\RR_+$ and polynomials of the
form \cite{tan2006stability}
$
     l(\xx)=\Sigma_{i=1}^n\epsilon_{i}x_i^{m},
$ 
where $\epsilon_{i}\in\RR_+$ and $m$ is assumed to be even. For
example, by using $\delta_{1}$ and $l_{1}(\xx)$, the second
condition in (\ref{condition:set}) can be relaxed as
$$\{\xx\in\RR^n: V(\xx)-1\leq 0, 
V(\xx)-l_1(\xx)+\delta_1\leq 0\}=\emptyset.$$
Therefore, the problem (\ref{condition:set}) can be transformed into the following
feasibility problem:
\begin{eqnarray*}
\left.\begin{array}{l@{}l}
    & \displaystyle \max_{V, \beta}\  \beta   \\
    & \text{s.t.} \, V(\bz)=0,\\
    & \,\,\,\{\xx\in\RR^n: V(\xx)-1\leq 0, V(\xx)-l_1(\xx)+\delta_1\leq 0\}=\emptyset,\\
    & \,\,\,\{\xx\in\RR^n: V(\xx)-1 \leq 0, -\frac{\partial V}{\partial \xx}\cdot\hat{\ff}(\xx,\theta)-l_2(\xx)+\delta_2\leq 0, \psi(\theta)\geq0\}=\emptyset,\,\,\forall u_{ij}\in\{\pm b_{ij}\}\\
    & \,\,\,\{\xx\in\RR^n: g(\xx)-\beta \leq 0, V(\xx)-
    1- \delta_3\geq 0\}=\emptyset.
\end{array}\right\}
\end{eqnarray*}

Suppose that $\Sigma[\xx]$ is the set of sum of squares (SOS)
polynomials in $\RR[\xx]$.
Since the constraints in the above problem 
involve
no equations and inequations, only a special case of Stengle's
Positivstellensatz
is needed,
as shown in the following corollary.
\begin{corollary}\label{cor:Positivstellensatz}
Let $F=\{f_i\}_{i=1,\dots,r}$ be a set of polynomials in $\RR[\xx]$.
The semi-algebraic set
$$
  \{\xx\in\RR^n: f_i(\xx)\geq0,i=1,\dots,r\}
$$
is empty if and only if there exist polynomials $s_0, s_1,\dots,
s_l\in \Sigma[\xx]$ such that
$$s_0+\sum_{j=1}^ls_jb_j=0$$
where $b_j\in\left\{f_1^{t_1}\cdots f_r^{t_r}: t_i\in \ZZ_{\geq 0}
\right\}$.
%
\end{corollary}


Applying Corollary \ref{cor:Positivstellensatz}, and removing all the crossing products of the involved
inequalities, we obtain the following
relaxed
problem:
\begin{eqnarray}\label{SOS:variable1}
\left.\begin{array}{l@{}l}
    & \displaystyle  \max_{V,\beta}\  \beta   \\
    &\text{s.t.} \, V(\bz)=0,\\
    &\quad
    \sigma_0(\xx)+\sigma_1(\xx)(1-V(\xx))+\sigma_2(\xx)(-V(\xx)+l_1(\xx)-\delta_1)=0,\\
    &\quad \lambda_0(\xx)+\lambda_1(\xx)(1-V(\xx))+\lambda_2(\xx)(\frac{\partial V}{\partial\xx}\cdot\hat{\ff,\theta}(\xx)+l_2(\xx)-\delta_2)=0,\\
    &\quad
    \rho_0(\xx)+\rho_1(\xx)(\beta-g(\xx))+\rho_2(\xx)(V(\xx)-1-\delta_3)=0,
\end{array}\right\}
\end{eqnarray}
with $\sigma_\iota(\xx),\lambda_\iota(\xx), \rho_\iota(\xx)\in
\Sigma[\xx], \iota=0, \dots, 2$ and for any $u_{ij}\in\{\pm b_{ij}\}$.

Suppose that the Lyapunov function $V(\xx)$ to be computed is a polynomial of degree $d$ and has the form
$
V(\xx)=\sum_{\alpha} c_{\alpha}\xx^{\alpha}$ where~$c_{\alpha} \in \RR$, $\xx^{\alpha}=x_{1}^{\alpha_{1}}\cdots x_{n}^{\alpha_{n}}$ and
$\alpha=(\alpha_{1},\ldots,\alpha_{n}) \in {\ZZ}_{\geq 0}^{n}$ with
$\sum_{i=1}^{n} \alpha_{i} \leq d$. The decision variables of the problem (\ref{SOS:variable1}) are
$\beta$ and the coefficients of all the unknown polynomials
occurred in (\ref{SOS:variable1}), such as $V(\xx),
\sigma_\iota(\xx),$ $\lambda_\iota(\xx)$ and $\rho_\iota(\xx)$.
Clearly, some nonlinear terms which are products of the undetermined
coefficients will occur in (\ref{SOS:variable1}),
which yields a non-convex bilinear matrix inequalities (BMI)
problem.
To solve BMI problems, either a Matlab package PENBMI
solver \cite{KS2005}, which combines the (exterior) penalty and
(interior) barrier method with the augmented Lagrangian method, can
be applied directly,
or an iterative method can be applied by fixing $\beta$
and the polynomials alternatively, which leads to a sequential
convex LMI problem. The reader can refer to  \cite{yang2012exact} for more details.


Remark that, the above proposed method is also applicable to computing the largest possible
estimate of the DOA
for a non-polynomial system (\ref{EQ:nonlinear}) at $\bz$ through a fixed Lyapunov function $V(\xx)$.
Let $\Omega_{c}:=\{\xx\in\RR^n: V(\xx)\leq c\}.$ We will compute
$$
\Omega_{c^*}=\{\xx\in\RR^n: V(\xx)\leq c^*\},
$$
where
$c^*=\sup\{c\in \RR: \dot{V}(\xx)<0, \,\,\forall\xx\in\Omega_c
\setminus\{\bz\}\}.
$
Due to the existence of non-polynomial terms, $c^*$ cannot be computed explicitly. Instead, we will compute
the lower and upper bounds of $c^*$ as follows.
Replacing the involved non-polynomial $\phi_{ij}$ by uncertain polynomials,
computing the lower bound $c_d$ of $c^*$ can be relaxed as the
following problem:
\begin{eqnarray}\label{ROA:fixed1}
\left.\begin{array}{l@{}l}
   &\displaystyle  c_d=\max_{c,\xx} \,\,  c    \\
    &\quad\text{s.t.} \,\, \frac{\partial V}{\partial
    \xx}\cdot \hat{\ff}(\xx)<0,\,\, \forall
    \xx\in\Omega_c\backslash\{\bz\},\,\,\forall u_{ij}\in\{\pm b_{ij}\},
\end{array}\right.
\end{eqnarray}
where $d$ is the degree of the interpolation polynomials.
Clearly,
$c_d$ will converge to $c^*$ when $d$ tends to $\infty$.
Next, we will search for a tight upper bound $\upsilon_d$ of $c^*$.
To achieve this, let us look for $\upsilon_d$ such that, for each
$u_{ij}\in\{\pm b_{ij}\}$, the constant semi-algebraic system
\begin{equation}\label{upper_bound}
    V(\xx)-\upsilon_d\leq 0 \wedge \xx\neq\bz\wedge\frac{\partial V}{\partial\xx}\cdot\hat{\ff}(\xx)\geq0
\end{equation}
has real solutions, which implies that $\Omega_{\upsilon_d}$ is not an estimate of the DOA.
%
%
Based on bisection, $\upsilon_d$ can be computed by Maple packages {\it RegularChains}, {\it
DISCOVERER} \cite{xia07} and {\it RAGLib}~\cite{mohab03}.

\section{Experiments}\label{sec:exp}
Let us present some examples of DOA analysis of non-polynomial 
systems.
\begin{example}\cite[Example 1]{chesi2009estimating} %
Consider a non-polynomial 
system
\begin{eqnarray*}
\left\{\begin{array}{l@{}l}
   \dot{x}_1=-x_1+x_2+0.5(e^{x_1}-1), \\
   \dot{x}_2=-x_1-x_2+x_1x_2+x_1\cos x_1.
   \end{array}\right.
\end{eqnarray*}
To estimate the DOA of this system,
we need to approximate the occurred non-polynomial terms $e^{x_1}$ and
$\cos x_1$ by uncertain polynomials.
Based on the
technique in Section \ref{sec:approx}, we obtain
\begin{eqnarray*}
\left\{\begin{array}{l@{}l}
    &e^{x_1}=1+(1.0000004+u_1)x_1+\cdots+0.0014482244x_1^6,\\
    &\cos x_1=1-(0.5+u_2)x_1^2+0.041669352x_1^4-0.0013878601x_1^6,\\
    &-0.6\leq x_1\leq 0.6,\\
    &-3\times10^{-6}\leq u_1\leq 3\times10^{-6},\\
    &-1.2\times10^{-6}\leq u_2\leq 1.2\times10^{-6},
\end{array}\right.
\end{eqnarray*}
and the associated uncertain polynomial system.

We first consider a fixed Lyapunov function $V(x_1,x_2)=x_1^2+x_2^2$.
For the given degree $6$ of interpolation polynomials, after solving the corresponding SOS programming (\ref{ROA:fixed1}),
we obtain
the lower bound $c_6=0.321064$ of $c^*$,
 which is an improvement over the results in
\cite{chesi2009estimating} with the lower bound 0.3210. Furthermore,
by solving the problem (\ref{upper_bound}) we obtain a tight
upper bound $\upsilon_6=0.3216$ of $c^*$.
\begin{figure}\label{fig_ex4}
\centering \subfigure[boundaries of $\Omega_{c_6}$ with
$c_6=0.321064$ and $\dot{V}=0$] {
      \includegraphics[width=0.3\textwidth,angle=270]{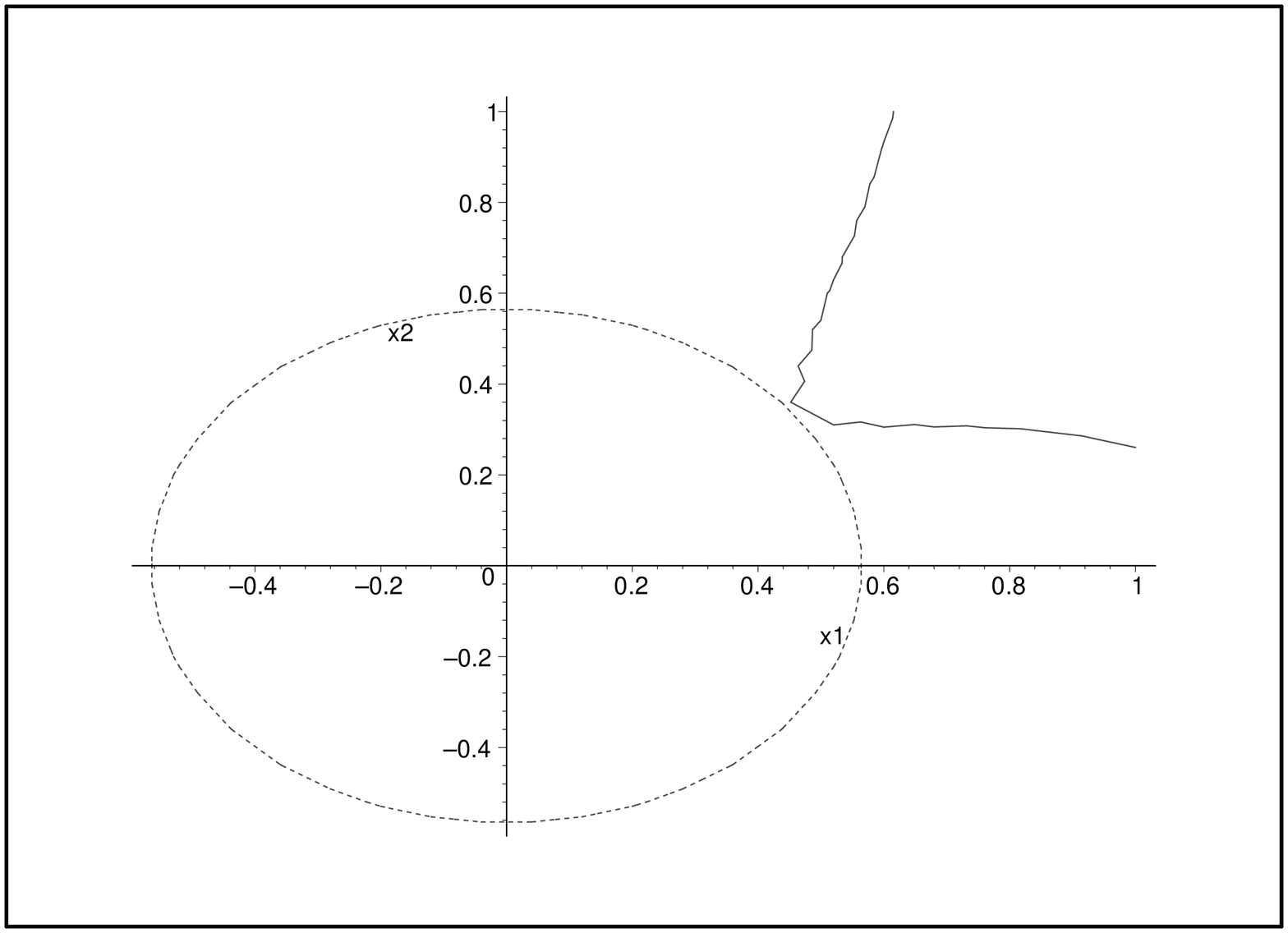}
      \label{fig_firstsub}
} \subfigure[boundaries of the estimates for $\deg{V=2}$ and
$\deg{V=4}$] {
      \includegraphics[width=0.3\textwidth,angle=270]{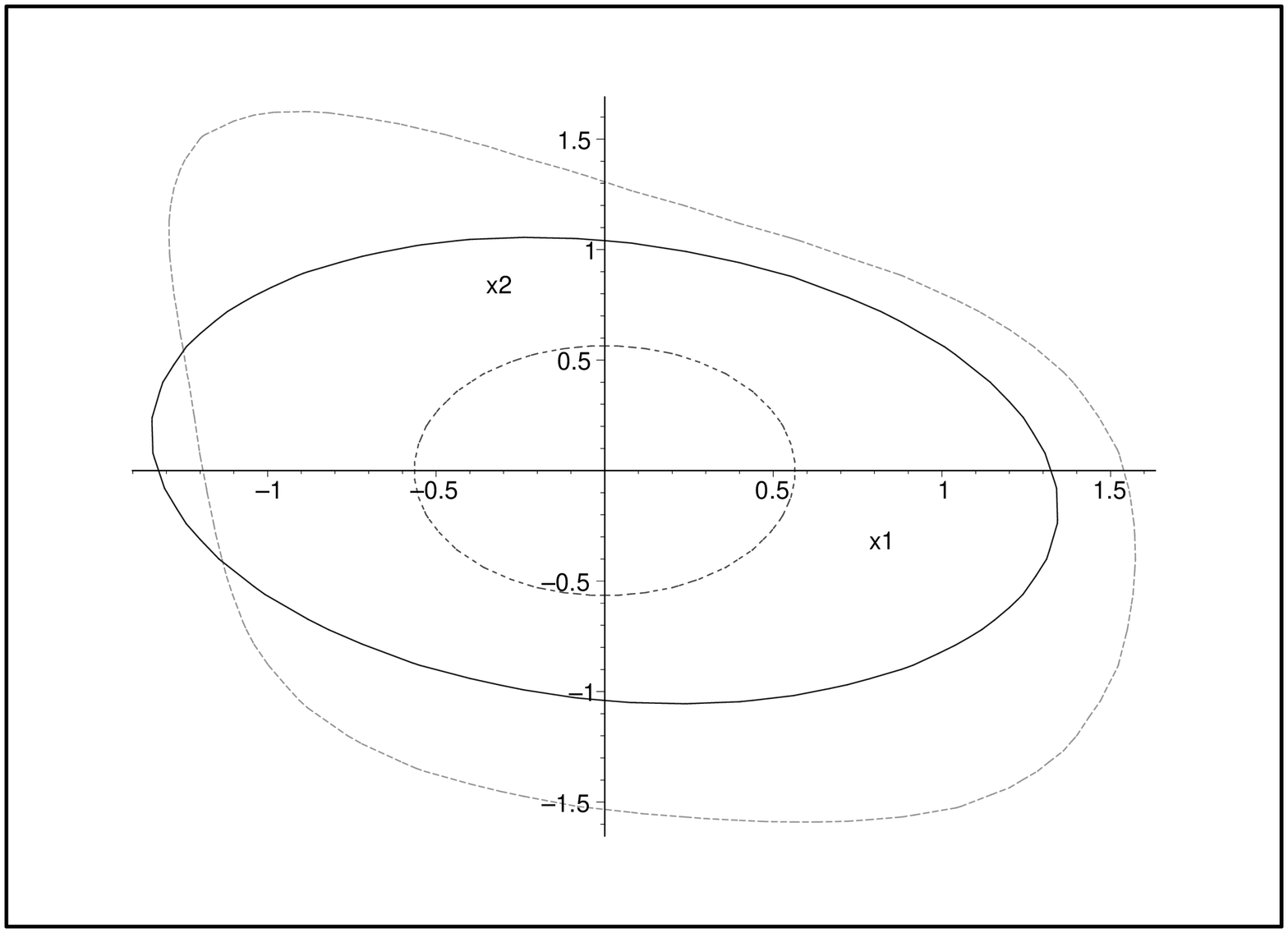}
      \label{fig_secondsub}
}  \caption{Results of Example 4.} \label{fig_subfigures}
\end{figure}

Next, we  estimate the DOA with variable
Lyapunov functions.
Suppose $g(x_1,x_2)=x_1^2+x_2^2$. For $\deg V=2$, solving
the SOS programming (\ref{SOS:variable1}) with BMI constraints yields
\begin{eqnarray*}
\begin{split}
    &{V}(x_1,x_2)=0.56678683x_1^2+0.23598133x_1x_2+0.92086339x_2^2,\\
    &{\beta}=1.0453916,
\end{split}
\end{eqnarray*}
which is an improvement over the result from
\cite{chesi2009estimating} where ${\beta}=1.0404$. Similarly, for
$\deg V=4$ solving the SOS programming
(\ref{SOS:variable1}) with BMI constraints yields
${\beta}=1.4001306$ and
\begin{equation*}
    {V}(x_1,x_2)=0.068693712x_1^2+\cdots+0.27723966x_1^4+0.2167918x_2^4,
\end{equation*}
which is an improvement over the result from
\cite{chesi2009estimating} where ${\beta}=1.2769$. Therefore,
$\Omega_{V}$ is an estimate of the DOA of the given system. Figure 2
shows the results obtained with Lyapunov functions of degrees 2 and
4. $\hfill \Box$
\end{example}

\begin{example}\cite[Example 2]{chesi2009estimating}%
Consider a non-polynomial 
system
\begin{eqnarray*}
\left\{\begin{array}{l@{}l}
   \dot{x}_1=x_2, \\
   \dot{x}_2=-0.2x_2+0.81\sin x_1 \cos x_1-\sin x_1.
   \end{array}\right.
\end{eqnarray*}
Using the
technique in Section \ref{sec:approx}, we obtain the approximations of the non-polynomial
terms $\sin x_1$ and $\cos x_1$ as follows
\begin{eqnarray*}
\left\{\begin{array}{l@{}l}
    &\sin 2x_1=(2+u_1)x_1-1.3333091x_1^3+0.26625372x_1^5-0.023889715x_1^7,\\
    &\sin x_1=(1+u_2)x_1-0.16666643x_1^3+0.008331760x_1^5-0.00019471928x_1^7,\\
    &-0.84\leq x_1\leq 0.84,\\
    &-5.4\cdot10^{-5}\leq u_1\leq 5.4\cdot10^{-5},\\
    &-1\cdot10^{-7}\leq u_2\leq
    1\cdot10^{-7},
\end{array}\right.
\end{eqnarray*}
and the associated uncertain polynomial system.

We first fix the Lyapunov function $V(x_1,x_2)=x_1^2+x_1x_2+4x_2^2$.
Let the degree of interpolation polynomials be 7. Solving the corresponding SOS programming (\ref{ROA:fixed1}), we
obtain the results for the lower bound $c_7=0.69922$ of $c^*$,
which is an improvement over the result from
\cite{chesi2009estimating} where the lower bound was 0.6990. Furthermore, by solving the problem (\ref{upper_bound}) we
obtain a tighter upper bound $\upsilon_7=0.6998$ of $c^*$.

We then estimate the DOA with variable
Lyapunov functions. Suppose $g(x_1,x_2)=x_1^2+x_2^2$. When
$\deg V=2$, by solving the SOS programming
(\ref{SOS:variable1}) with BMI constraints, we obtain
\begin{eqnarray*}
\begin{split}
    &{V}(x_1,x_2)=1.01636667x_1^2+0.84993333x_1x_2+3.40233333x_2^2,\\
    &{\beta}=0.287706,
\end{split}
\end{eqnarray*}
which is an improvement over the result from
\cite{chesi2009estimating} where ${\beta}=0.2809$. Similarly, when
$\deg V=4$, by solving the SOS programming (\ref{SOS:variable1})
with BMI constraints, we obtain ${\beta}=1.92156$ and
\begin{equation*}
    {V}(x_1,x_2)= 0.053849691x_1^2+\cdots+0.11144243x_1^4+0.058855229x_2^4,
\end{equation*}
which is an improvement over the result from
\cite{chesi2009estimating} where ${\beta}=1.1236$. Therefore,
$\Omega_{\widetilde{V}}$ is an estimate of the DOA of the given
system. Figure 3 shows the results obtained with Lyapunov functions
of degrees 2 and 4.
\begin{figure}\label{fig_ex5}
\centering \subfigure[boundaries of $\Omega_{c_7}$ with
$c_7=0.69922$ and $\dot{V}=0$] {
      \includegraphics[width=0.3\textwidth,angle=270]{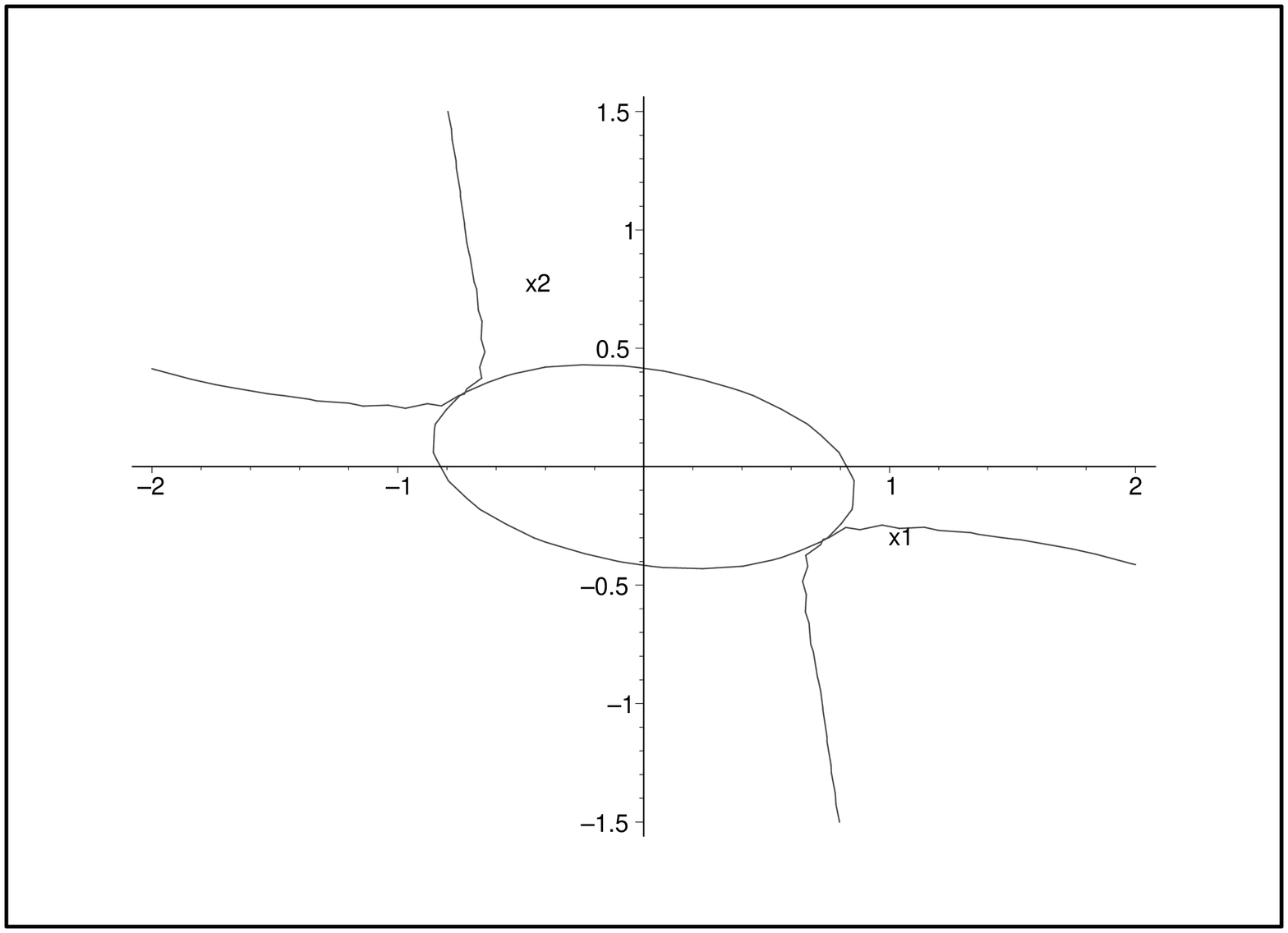}
      \label{fig_firstsub}
} \subfigure[boundaries of the estimates for $\deg{V=2}$ and
$\deg{V=4}$] {
      \includegraphics[width=0.3\textwidth,angle=270]{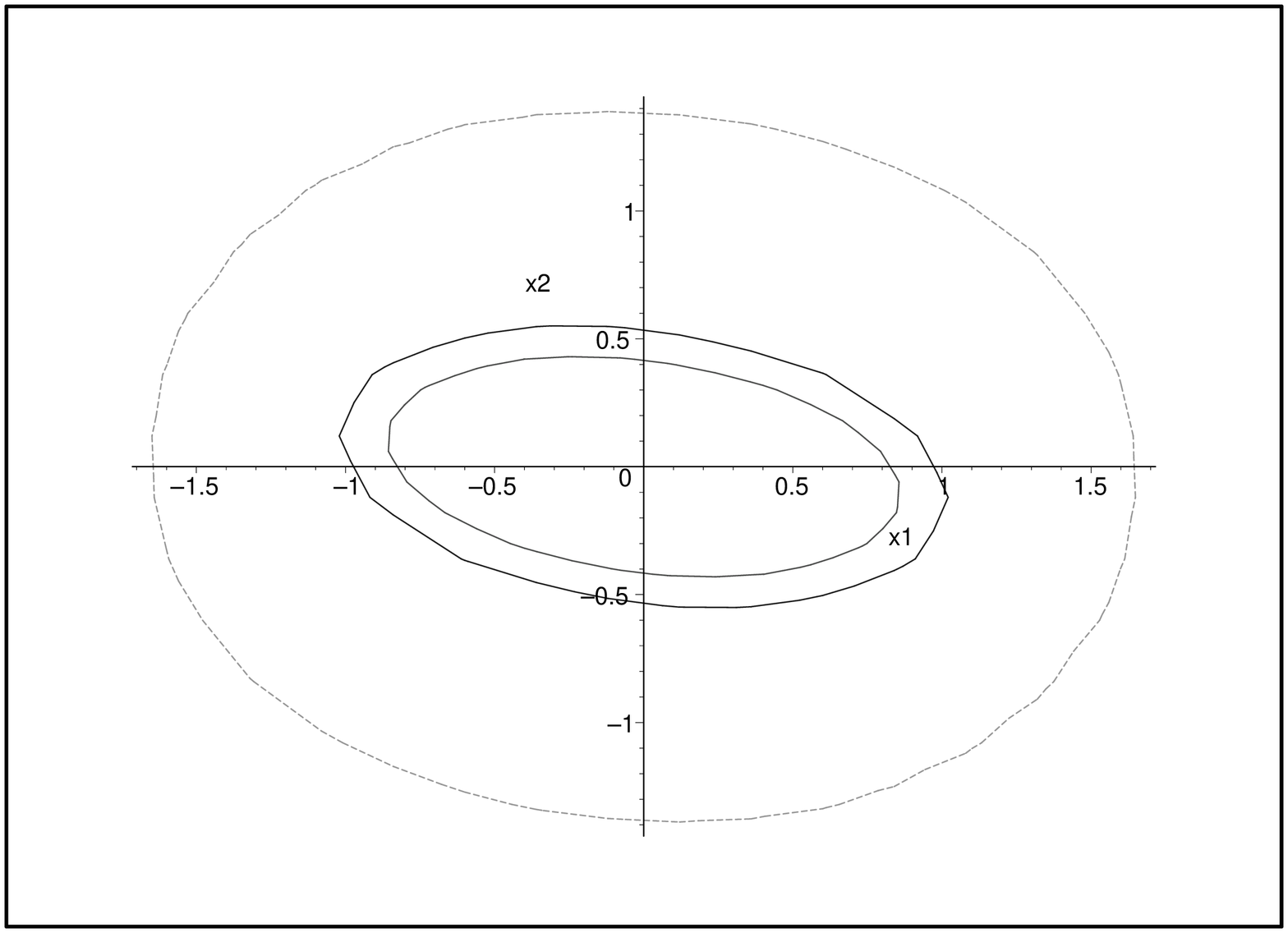}
      \label{fig_secondsub}
}  \caption{Results of Example 5.} \label{fig_subfigures}
\end{figure}
$\hfill \Box$
\end{example}

\begin{example}\cite[Example 2.2]{khalil2002nonlinear}
Consider an uncertain non-polynomial 
system
\begin{eqnarray*}
\left\{\begin{array}{l@{}l}
   \dot{x}_1=x_2, \\
   \dot{x}_2=-\theta x_2-10\sin x_1,
   \end{array}\right.
\end{eqnarray*}
for $0.2\leq \theta\leq 1$.
Based on the technique in Section \ref{sec:approx}, we obtain an approximation of the non-polynomial
term $\sin x_1$ as follows
\begin{eqnarray*}
\left\{\begin{array}{l@{}l}
    &\sin x_1=(1+u_1)x_1-0.1666426901x_1^3+0.008282118073x_1^5-0.0001751721223x_1^7,\\
    &-2.4\leq x_1\leq 2.4,\\
    &-4.365606\times10^{-4}\leq u_1\leq
    4.365606\times10^{-4},
\end{array}\right.
\end{eqnarray*}
and the associated uncertain polynomial system.

Suppose $g(x_1,x_2)=x_1^2+x_2^2$. For $\deg V=4$, solving the
SOS programming (\ref{SOS:variable1}) with BMI constraints yields
${\beta}=0.66552836$ and
\begin{equation*}
    {V}(x_1,x_2)= 1.1629845x_1^{2}+ \cdots+ 0.51014802x_1^{4}+ 0.010528x_2^{4}.
\end{equation*}
Then $\Omega_{{V}}$ is an estimate of the DOA of the given
system. $\hfill \Box$
\end{example}

\section{Conclusion}\label{sect:conclusion}

In this paper, we present a method on stability region analysis of
non-polynomial systems via Lyapunov functions. A
polynomial approximation technique, based on multivariate polynomial
interpolation and error analysis, is applied to compute an uncertain
polynomial system, whose set of trajectories contains that of the original non-polynomial
system. To estimate DOA of the uncertain polynomial system, we apply
Positivstellensatz to transform polynomial optimization problem
into the corresponding (bilinear) sum of
squares programming, which can be solved using the PENBMI solver
or iterative method. Experiments on the benchmark non-polynomial
systems 
show that our approach provides better estimates.

\end{document}